\documentclass[preprintnumbers,amsmath,amssymb,twocolumn,pra,showpacs,nofootinbib]{revtex4-1}

\usepackage{mathtools}
\usepackage{algorithm}
\usepackage{algpseudocode}
\usepackage[roman]{sublabel}
\usepackage{setspace}
\usepackage{graphicx}
\usepackage{dcolumn}
\usepackage{amssymb}
\usepackage{color}
\usepackage{comment}
\usepackage{epstopdf}
\usepackage{bm}
\usepackage{enumerate}
\usepackage{fancyhdr,color}
\usepackage{verbatim}
\usepackage{amsmath,amsthm}
\usepackage{amsfonts}
\usepackage{hyperref}
\usepackage{xparse}
\usepackage{physics}
\usepackage[normalem]{ulem}

\def\eabs{\epsilon_{\rm abs}}
\def\erel{\epsilon_{\rm rel}}

\def\cH{\mathcal{H}}

\def\cX{\mathcal{X}}
\def\cY{\mathcal{Y}}
\def\cZ{\mathcal{Z}}

\def\Tr{\mathrm{Tr}}
\def\Pr{\mathrm{Pr}}

\def\poly{\mathrm{poly}}
\newtheorem{definition}{Definition}
\newtheorem{theorem}{Theorem}
\newtheorem{lemma}{Lemma}
\newtheorem{corollary}{Corollary}

\def\one{{\mathchoice {\rm 1\mskip-4mu l} {\rm 1\mskip-4mu l} {\rm
1\mskip-4.5mu l} {\rm 1\mskip-5mu l}}}
\def\dqc1{\textsc{DQC1}}

\newcommand{\sket}[1]{| #1\rangle}  
\newcommand{\sbra}[1]{\langle #1|}  

\renewcommand{\Re}{\mathrm{Re}}

\algnewcommand\algorithmicto{\textbf{to}}
\algrenewtext{For}[3]{\algorithmicfor\ $#1 \gets #2$ \algorithmicto\ $#3$ \algorithmicdo}

\begin{document}
    \title{Computing partition functions in the one clean qubit model}

 \author{Anirban N. Chowdhury}
\affiliation{D\'epartement de Physique \& Institut Quantique, Universit\'e de Sherbrooke, QC J1K 2R1, Canada}

\author{
Rolando D. Somma}
\affiliation{Theoretical Division, Los Alamos National Laboratory, Los Alamos, NM 87545, USA}

\author{Yi\u{g}it Suba\c{s}\i}
\affiliation{Computer, Computational, and Statistical Sciences Division, Los Alamos National Laboratory, Los Alamos, NM 87545, USA}

\date{\today}
    
\begin{abstract}
We present a method to approximate partition functions
of quantum systems using mixed-state quantum computation. For positive semi-definite Hamiltonians,
our method has expected running-time that is almost linear in $(M/(\erel \cZ ))^2$,
where $M$ is the dimension of the quantum system, $\cZ$ is the  partition function,
and $\erel$ is the relative precision. It is based on approximations of
the exponential operator as linear combinations of certain operators related to 
block-encoding of Hamiltonians or Hamiltonian evolutions.
The trace of each operator is estimated using a standard algorithm
in the one clean qubit model. For large values of
$\cZ$, our method may run faster than exact classical methods, 
whose complexities are polynomial in $M$. We also prove
that a version of the partition function estimation problem within
additive error is complete for the so-called DQC1 complexity class, suggesting that our method
provides a super-polynomial speedup for certain parameter values.
To attain a desired relative precision, we develop
a classical procedure based on a sequence of approximations
within predetermined additive errors that may be of independent interest.
\end{abstract}

\maketitle
    
\section{Introduction}
    
One of the most important quantities used to describe a physical system in thermodynamic equilibrium is the partition function $\cZ$. Many thermodynamic properties, such as the free energy or entropy, can then be derived from $\cZ$ using simple mathematical relations~\cite{Pat72}. Partition functions also appear naturally in many other problems in mathematics and computer science, such as counting the solutions of constraint satisfaction problems~\cite{BG05}. Therefore, developing novel algorithms for partition functions is of great importance~\cite{Bra08,PW09,NDR+09,CS16,HMS19}.

In this paper, we present a method
for approximating partition functions using
the one clean qubit model of computation.
In this model, one qubit is initialized in a pure state,
in addition to $n$ qubits in the maximally mixed state~\cite{KL98}.
This is in contrast to standard quantum computation,
where the many-qubit initial state is pure~\cite{NC01}.
The one clean qubit model has attracted significant attention as it appears
that some problems can be solved efficiently within this model for which no efficient classical algorithms are known to exist~\cite{PBL+04,SJ08,CM18}. Additionally, this model is practically relevant for mixed-state quantum computation, e.g., it is suitable to describe liquid-state NMR~\cite{LKC+02,NSO+05}, and for quantum metrology~\cite{BS08}. We demonstrate further advantages of the one clean qubit model by describing a method to estimate partition functions of quantum systems.

Our main result is a method that outputs an estimate $\hat \cZ$ of $\cZ$ within given relative precision $\erel >0$ and with high success probability $(1-\delta)<1$. 
To achieve this goal we first give algorithms in the one clean qubit model that can estimate partition functions of certain $m$-qubit systems within a desired additive error.
The basic idea behind these algorithms is the fact that $\cZ$
can be approximated from linear combinations of the traces of 
certain unitary operators. We present two approaches: one in which the unitaries are constructed from a block-encoding of the Hamiltonian and another in which they are constructed from Hamiltonian evolutions.\footnote{Whether one approach is more suitable than the other will depend on the specification of the Hamiltonian of the system. See Sections~\ref{sec:LCU}, \ref{sec:Uimp}, and \ref{sec:complexity} for details.} The trace of each such unitary can be estimated by repeated uses of the well-known trace-estimation algorithm of Fig.~\ref{fig:DQC1_Alg}, which allows us to compute $\mathcal{Z}$ with additive error. Finally, we obtain an estimate of $\cZ$ within a desired relative error by iterating multiple additive-error estimations.

For positive semi-definite Hamiltonians, our algorithm for obtaining relative-error estimates of the partition function has expected running-time that is almost linear in $(M/(\erel \cZ ))^2$, where $M=2^m$.
In contrast, the runtime of a classical method based on exact diagonalization scales as $M^3$ and the runtime of the kernel polynomial method scales as $M$ \cite{weisse2006kernel}. 
Our method can thus provide a significant (super-polynomial) speedup in cases where $\cZ$, or the temperature, is sufficiently large. However, the improvement is less pronounced or may be lost when $\cZ$ is small, which corresponds to the low-temperature regime. In this regime, estimating the partition function is known to be computationally difficult~\cite{Sly10,HMS19}. Further, numerous complexity theoretic results effectively rule out the possibility of having efficient general-purpose algorithms for estimating the partition function. Exactly computing partition functions of classical systems, e.g., the Ising model is \#P-hard~\cite{JS93}, and even approximating them to a multiplicative error can be NP-hard; see Ref.~\cite{MQ13} for some results. While efficient polynomial-time approximation algorithms exist for high-temperature partition functions, these algorithms have not been shown to work at low temperatures, i.e., temperatures below the critical point~\cite{HMS19}. Our method as well as known quantum algorithms such as Ref.~\cite{PW09} can work at any temperature, even though the complexities may not always be favorable.

The complexity of additive-error estimates of the partition function has also received attention, notably in the context of quantum computation. It is known that for certain lattice models this problem can be either BQP-complete or DQC1-complete, but only in complex parameter regimes (complex temperatures) that do not correspond to physical scenarios~\cite{DDVM11}. In a physically relevant setting (real temperature), Brand\~{a}o showed that it is DQC1-hard to estimate the normalized partition function of certain $\log m$-local Hamiltonians within $1/\poly(m)$ additive error at temperature that is $\Omega(1/\poly(m))$~\cite{Bra08}. Our results imply that this problem is in fact DQC1-complete. This suggests that in some cases our algorithm provides a super-polynomial speed-up.

 \begin{figure}[htb]
    \includegraphics[width=5.5cm]{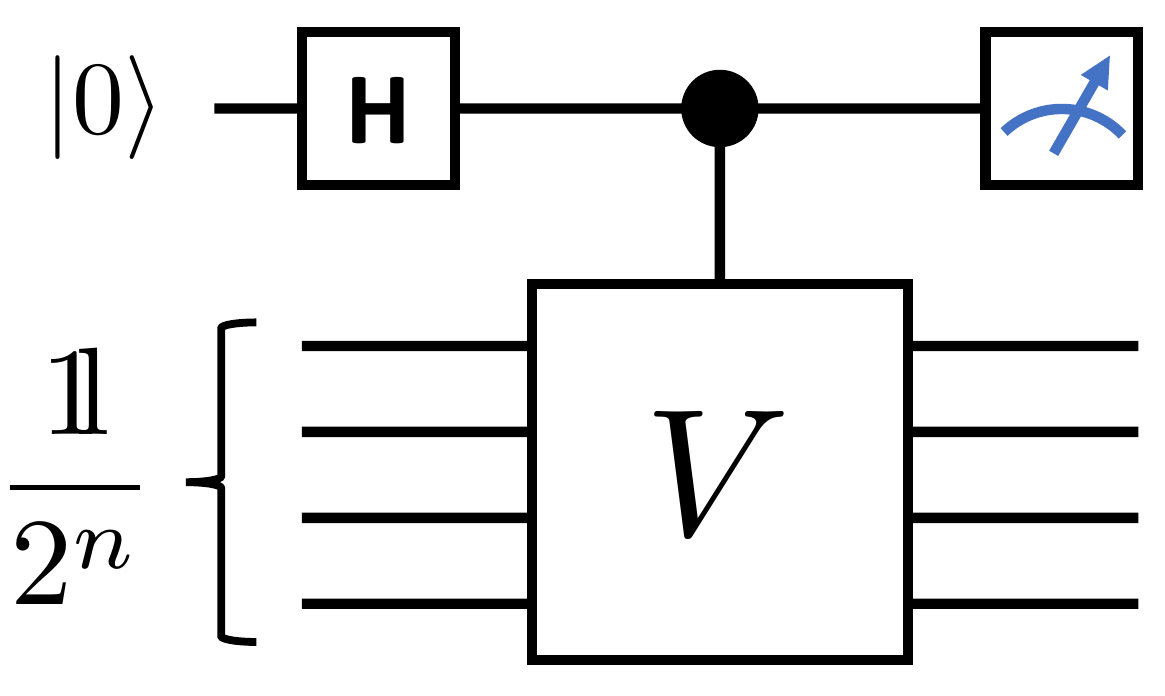}
     \caption{The trace-estimation algorithm to compute
     the renormalized trace of a unitary operator $V$.
     The clean (ancilla) qubit is initialized in the state $\ket 0$
     and then acted on with the Hadamard gate ${\rm H}$. The remaining $n$
     qubits are initialized in the completely mixed state $\one/2^n$.
     The filled circle denotes that $V$ is applied conditional on
     the state of the ancilla being $\ket 1$. Repeated projective measurements of the ancilla-qubit Pauli
     operator $\sigma_x$, resulting in $\pm 1$ outcomes, provide an estimate of the expectation
     $\langle \sigma_x \rangle=\Re\Tr[V]/2^n$.}
    \label{fig:DQC1_Alg}
\end{figure}

The rest of the paper is organized as follows.
In Sec.~\ref{sec:PFP} we state our assumptions and introduce the partition function problem
where the goal is to estimate $\cZ$ within given relative error.
In Sec.~\ref{sec:DQC1} we describe the one clean qubit model
and introduce the DQC1 complexity class. In Sec.~\ref{sec:LCU}
we provide two approximations to the exponential operator as specific linear combinations of unitaries, 
which will be used by our method. 
The unitaries in the approximation are related to block-encoding of Hamiltonians or Hamiltonian simulation, and
we describe the implementations of these in Sec.~\ref{sec:Uimp}. 
In Sec.~\ref{sec:alg}
we provide our main algorithms. In Sec.~\ref{sec:CC} 
we demonstrate the correctness of our method and in Sec.~\ref{sec:complexity} we 
establish its complexity. 
In Sec.~\ref{sec:complete} we show that a version of the partition function problem is DQC1-complete 
and we conclude in Sec.~\ref{sec:conc}.

We give some technical proofs in Appendices~\ref{app:prooflemma1},~\ref{app:chebyapprox}, and~\ref{app:unitapprox}. In Appendix~\ref{appx:est_rel_err} we develop a classical
procedure to estimate quantities with a given relative error and success probability, from estimations with suitable additive errors and success probabilities. This procedure is formulated under fairly general assumptions and it can be applied to a wide range of problems beyond the one considered in this paper.
Finally, in Appendix~\ref{app:rel_err_complexity} we bound the complexity of our algorithms for obtaining relative estimates to the partition function.

\section{Problem statement}
\label{sec:PFP}
        
We consider a discrete, $M$-dimensional quantum system with Hamiltonian $H$.
In the canonical ensemble, the partition function is 
\begin{align}
\label{eq:PF}
\mathcal{Z} := \text{Tr}\left(e^{-\beta H} \right) \;,
\end{align}
where $\beta \ge 0$ is the inverse temperature. That is, $\beta=1/k_{\rm B}T$, with $k_{\rm B}$ the Boltzmann constant and $T$ the temperature. For simplicity, we will focus on systems composed of $m$ qubits,
where $m=\log_2(M)$. Nevertheless, if one is interested in partition functions of quantum systems obeying
different particle statistics, such as bosonic or fermionic systems, the results in Refs.~\cite{SOGKL02,SOKG03} may be used to represent the corresponding operators in terms
of Pauli operators acting on qubits. The techniques developed here can then be used to study such systems.
Formally, we define the partition function problem (PFP) as follows:
\begin{definition}[PFP]
Given a Hamiltonian $H$, an inverse temperature $\beta\ge 0$,
a relative precision parameter $\erel>0$, and a probability of error $\delta>0$,
the goal is to output a positive number $\hat \cZ$ such that
\begin{align}
    \label{eq:PFP}
    \left |\hat \cZ  - {\cZ} \right | \le \erel {\cZ} \; ,
\end{align}
with probability at least $(1-\delta)$.
\end{definition}
The reason why we focus on relative approximations of the partition function is because they translate
to additive approximations for the estimation of {extensive} thermodynamic quantities such as entropy and free energy. For example, the free energy in thermodynamic equilibrium is given by $F=-(1/\beta) \log Z$. Using the estimate $\hat \cZ$ to obtain an estimate $\hat F$, we obtain $|F-\hat F|=O(\erel/\beta)$.
We will also consider additive approximations of $\cZ$ in our discussion\,---\,this is in fact the partition function problem studied in Ref.~\cite{Bra08}, for which a quantum algorithm in the circuit model is given. We show in Sec.~\ref{sec:alg} that the two problems are related.

Our main goal is to provide an algorithm that uses the one clean qubit model to solve the PFP. We will focus on Hamiltonians that have the form $H=\sum_{l=1}^L\alpha_l H_l$, $\alpha_l>0$,  $L=O(\text{poly}(m))$, and where each $H_l$ is either a unitary operator or a projector. We require that there exist efficient quantum circuits to implement either each $H_l$ (when it is unitary) or a unitary related to each $H_l$ (when it is a projector), as explained in Sec.~\ref{sec:Uimp}. Defining $\alpha=\sum_{l=1}^L\alpha_l$, we work with the renormalized Hamiltonian $H\gets H/\alpha$ and rescaled inverse temperature $\beta \gets \beta\alpha$ in order to simplify notation. The complexities of our algorithms depend implicitly on $\alpha$ through their dependence on the inverse temperature $\beta$.

\section{The one clean qubit model}
\label{sec:DQC1}
In the one clean qubit model, the initial state (density matrix) of a system of $n+1$ qubits is
\begin{align}
    \rho_i = \ketbra 0 \otimes \frac \one {2^n} \;,
\end{align}
where $\one$ is the identity operator over $n$ qubits. We write $\cH_n \equiv \mathbb C^{2^n}$ for the Hilbert space
associated with $n$-qubit quantum states. A quantum circuit
$U=U_{T-1} \ldots U_0$ is then applied to $\rho_i$, where
each $U_j$ is a two-qubit quantum gate, and a projective measurement
is performed
on the ancilla at the end. The outcome probabilities are $p_0$ and $p_1=1-p_0$,
where
\begin{align}
    p_0 = \Tr [(U \rho_i U^\dagger) (\ketbra 0 \otimes \one)] \; .
\end{align}

The complexity class DQC1 consists of decision problems
that can be solved within the one clean qubit model in polynomial time (in the problem size $s$)
with correctness probability $\ge 2/3$. We are allowed to act on $\rho_i$ with quantum circuits of length $\text{poly}(s)$, measure the ancilla, and repeat this $\poly(s)$ many times. In our definition, DQC1 contains
the class BPP, that is, the class of problems that can be solved in time $\text{poly}(s)$ using a classical computer (probabilistic Turing machine).

Remarkably, it can be shown that the problem of estimating $\Re \Tr[V]/2^n$ within additive error $\Omega(1/\text{poly}(n))$, where
$V$ is a quantum circuit of length $O(\text{poly}(n))$ acting on $\cH_n$,
is complete for the DQC1 class~\cite{KL98}. That is, any other problem in DQC1 
can be reduced to trace estimation. While this
is not a decision problem, it can be transformed to one by simple
manipulations~\cite{Shep06}. In this paper, however, we will mainly
focus on problems that can be reduced to trace estimation but
where the number of operations or steps are sometimes exponentially
large in $n$; that is, problems that are not necessarily in DQC1.

The trace-estimation algorithm 
is shown in Fig.~\ref{fig:DQC1_Alg}.
For a given quantum circuit $V$,
the quantum state before measurement
is
\begin{align}
\nonumber
    \rho_f = \frac 1 {2^{n+1}} &( \ketbra 0 \otimes \one + \ketbra{0}{1} \otimes V^\dagger +  \\
      & + \ketbra{1}{0} \otimes V + \ketbra 1 \otimes \one ) \;.
\end{align}
Projective measurements on $\rho_f$ of the ancilla Pauli operator $\sigma_x$
result in $\pm 1$ outcomes whose average is an estimator of $\Re \Tr [V]/2^n$.
In particular,
\begin{align}
    \langle \sigma_x \rangle &:= \Tr[\rho_f\, \sigma_x] \\
    & = \frac {\Re \Tr[V]}{2^n} \;.
\end{align}
We will then estimate the expectation $\langle \sigma_x \rangle$, and thus $\Re \Tr[V]$,
from finitely many uses of the trace-estimation algorithm.
We obtain:
\begin{lemma}
\label{lem:Hoef}
Given $\varepsilon>0$, $\delta_0>0$,
and a quantum circuit $V$ acting on $n$ qubits, we can obtain an estimate $\hat \xi_V \in \mathbb R$
that satisfies 
\begin{align}
    \left | \hat \xi_V - \Re \Tr[V]\right | \le \varepsilon \;,
\end{align}
with probability at least  $(1-\delta_0)$, using the trace-estimation algorithm
$Q = \lceil (2^{2n+1}/\varepsilon^2) \log(2/\delta_0) \rceil$ times.
\end{lemma}

The proof of Lemma~\ref{lem:Hoef} is a simple consequence of Hoeffding's inequality~\cite{Hoe63}
and is given in Appendix~\ref{app:prooflemma1}.
If $s_x$ is the average of the measurement outcomes of $\sigma_x$, the estimate
is simply $\hat \xi_V = 2^n s_x$.

For our method, we will be interested in estimating the
trace of a given block of a unitary matrix within given additive error.
More specifically, let $W$ be a quantum circuit
defined on a system of $m+m'$ qubits. We obtain:
\begin{corollary}
\label{cor:Hoef2}
Given $\varepsilon>0$, $\delta_0>0$,
and a quantum circuit $W$ acting on $m+m'$ qubits, we can obtain an estimate $\hat \chi_W \in \mathbb R$
that satisfies 
\begin{align}
    \left | \hat \chi_W - \Re \Tr [\bra 0_{m'}W \ket 0_{m'}]\right | \le \varepsilon \;,
\end{align}
with probability at least $(1-\delta_0)$, using the trace-estimation algorithm
$Q = \lceil (2^{2(m+m')+1}/\varepsilon^2) \log(2/\delta_0) \rceil$ times. Here, $\ket 0_{m'} \in \cH_{m'}$
is the zero state of $m'$ qubits and $\bra 0_{m'}W \ket 0_{m'}$ is the corresponding block of $W$.
\end{corollary}

Corollary~\ref{cor:Hoef2} follows from the observation that
there is a quantum circuit $V$, acting on $n=m+2m'$ qubits,
and
\begin{align}
\label{eq:blocktrace}
    \frac 1 {2^{m'}}\Tr [V] = \Tr [\bra 0_{m'}W \ket 0_{m'}] \;;
\end{align}
see Ref.~\cite{SJ08}.
The unitary $V$ is described in Fig~\ref{fig:DQCkb_Alg}. The proof of Cor.~\ref{cor:Hoef2}
follows from Lemma~\ref{lem:Hoef},
where the number of qubits is $n=m+2m'$. The estimate in this case is $\hat \chi_W = 2^{m+m'} s_x$.

 \begin{figure}[htb]
    \includegraphics[width=6cm]{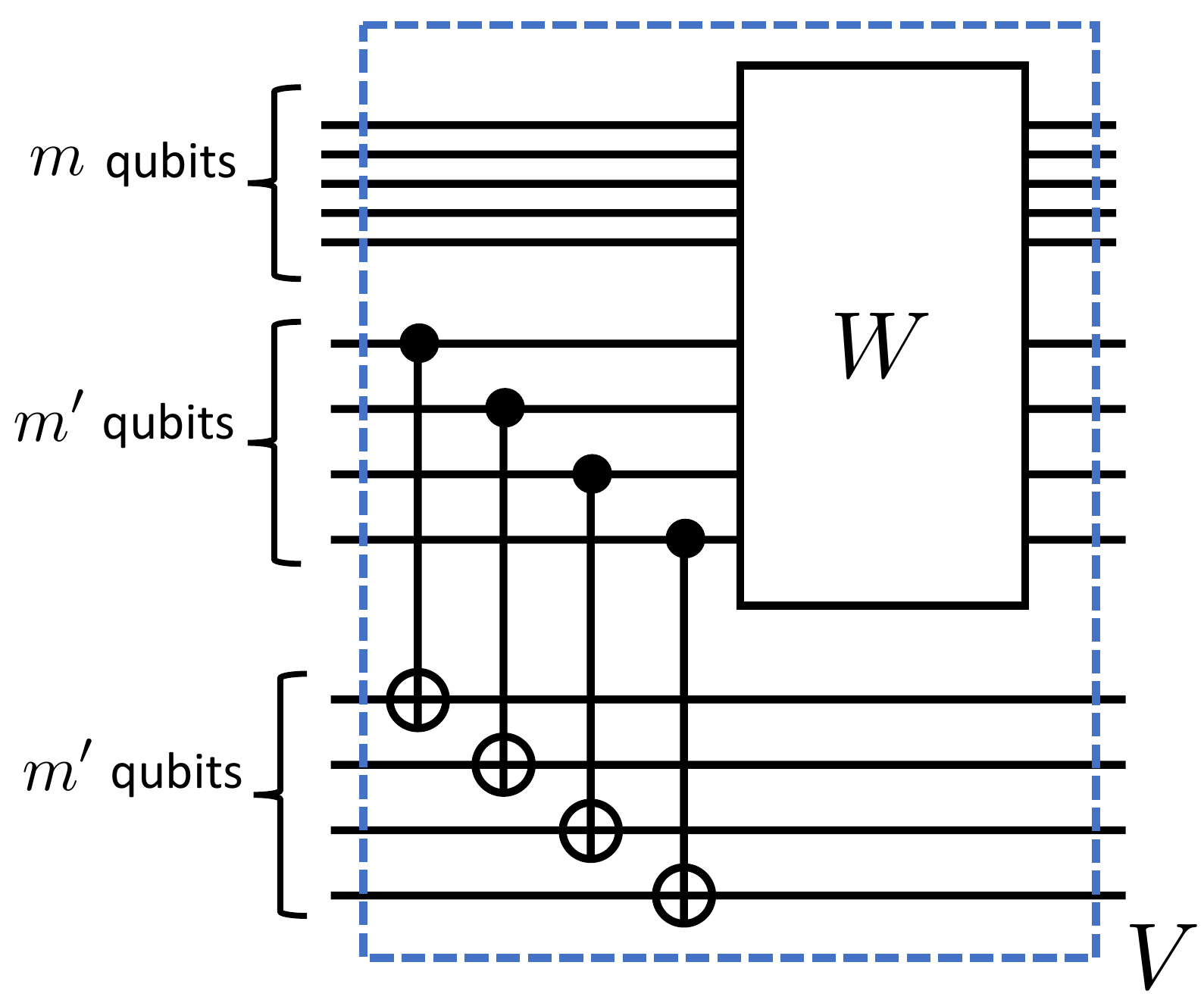}
     \caption{Quantum circuit $V$ that satisfies $\frac 1 {2^{m'}}\Tr[V]=\Tr [\bra 0_{m'}W \ket 0_{m'}]$.
     The first operations are a sequence of $m'$ CNOT gates on the corresponding pairs of ancilla qubits~\cite{SJ08}.
     }
    \label{fig:DQCkb_Alg}
\end{figure}

\section{Approximations of the exponential operator}
\label{sec:LCU}

Our method for estimating the partition function in the one clean 
qubit model proceeds by approximating it as a weighted sum of traces of unitary operators.
Each such trace can then be computed through repeated uses of
the trace-estimation algorithm of Fig.~\ref{fig:DQC1_Alg}.
We now describe two approximations of the exponential operator 
that will be used.

\subsection{Chebyshev approximation}
\label{sec:chebyshev}
The first approximation is based on Chebyshev polynomials. 
If the $m$-qubit Hamiltonian  satisfies $\|H\| \le 1$, we obtain
\begin{align}
\label{eq:chebdec}
    e^{-\beta H} = \sum_{k=-\infty}^\infty (-1)^k I_k(\beta) T_k(H) \;.
\end{align}
Here, $I_k(x) \in \mathbb R$ are the modified Bessel functions of the first kind. $T_k(H)$ is an operator acting on $\cH_m$
obtained by replacing $x$ by $H$ in $T_k(x)$,
the $k$-th Chebyshev polynomial of the first kind\,---\,see Appendix~\ref{app:chebyapprox}. 
We will approximate the exponential operator by a finite sum,
by noticing that $I_k(\beta)$ decays exponentially fast in the large $k$ limit 
(for fixed $\beta$). In Appendix~\ref{app:chebyapprox} we show:
\begin{lemma}
\label{lem:finitecheby}
Given $\eabs>0$ and $\beta \ge 0$, we can choose $K = \lceil m+ e \beta + \log_2(1/\eabs) +2 \rceil$ such that
\begin{align}
    \|S_K - e^{-\beta H} \|_1 \le \eabs/2 \; ,
\end{align}
where $\|\cdot\|_1$ denotes the trace norm and
\begin{align}
\label{eq:chebyappx}
    S_K &:= \sum_{k=-K}^K (-1)^k I_k(\beta) T_k(H) \; .
\end{align}
\end{lemma}

Equivalently, $S_K=I_0(\beta) + 2 \sum_{k=1}^K (-1)^k I_k(\beta) T_k(H)$.
To represent $S_K$ as a linear combination
of suitable operations for our method, 
we further assume that there exist unitary operators $W_H$ and $\tilde G$,
acting on $m+m'$ qubits, that satisfy
\begin{align}
\label{eq:iterate}
    T_k(H) = \sbra{0}_{m'}\tilde G^\dagger (W_H)^k \tilde G\sket{0}_{m'} \, .
\end{align}
The operation $W_H$ in Eq.~\eqref{eq:iterate}
is the ``unitary iterate'' or the quantum walk operator as used recently for Hamiltonian simulation~\cite{LC19}   or linear algebra problems~\cite{CKS17}, and $\tilde G$ is a related state-preparation unitary. 
We describe $W_H$ and $\tilde G$ in detail in Sec.~\ref{sec:blockencoding}.

Our first approach solves the PFP using the relation
\begin{align}
 \label{eq:firstappx}
    \Tr[S_K] =  I_0(\beta)&M+2 \sum_{k=1}^K (-1)^kI_k(\beta)  \nonumber \\
    &\times\Re\Tr [\bra 0 _{m'}\tilde G^\dagger (W_H)^k \tilde G \ket 0_{m'}] \;,
\end{align}
which is an $\eabs/2$ approximation to $\cZ$.
We can use the construction in Fig.~\ref{fig:DQCkb_Alg} in the trace-estimation algorithm of Fig.~\ref{fig:DQC1_Alg} to obtain $\Re\Tr [\bra 0 _{m'}\tilde G^\dagger (W_H)^k \tilde G \ket 0_{m'}]$.

\subsection{LCU approximation}
\label{sec:lcu}
The second approximation is based on the so-called Hubbard-Stratonovich transformation~\cite{Hub59,CS16}. If $H \ge 0$,
we obtain
\begin{align}
\label{eq:HST}
    e^{-\beta H} = \frac 1{\sqrt{2 \pi}} \int dy \; e^{-y^2/2} e^{-i y \sqrt{2 \beta H}} \;.
\end{align}
Here, $\sqrt H$ is also a Hermitian operator that refers to one of the square roots of $H$.
As the eigenvalues of $H$ are non-negative, this case appears to be more restrictive.
Nevertheless, the assumption $H\ge 0$ may be met after a simple pre-processing step that shifts $H$\,---\,see Sec.~\ref{sec:Uimp}.

We wish to obtain an approximation of $e^{-\beta H}$ by a finite linear
combination of unitaries following Eq.~\eqref{eq:HST}. This approximation is analyzed in Appendix~\ref{app:unitapprox} and
was also studied in Ref.~\cite{CS16}. 
If $\|H \| \le 1$, we obtain: 
 \begin{lemma} 
 \label{lem:discreteHTS}
 Given $\eabs > 0$ and $\beta \ge 0$, we can choose
 \begin{align}
      J &= \lceil 12 (\sqrt \beta + \sqrt{m +\log_2(1/\eabs)})\sqrt{m +\log_2(1/\eabs)}\rceil
 \end{align}
 and
 \begin{align}
     \delta y = \left( 2 ( \sqrt {\beta} + \sqrt{ m+\log_2(1/\eabs)} )\right)^{-1}
 \end{align}
 such that
 \begin{align}
    \left \| X_J - e^{-\beta H} \right \|_1 \le \eabs/4 \;,
\end{align}
where
 \begin{align} 
\label{eq:discreteHTS1}
X_J := \frac{\delta y}{\sqrt{2\pi}}\sum_{j=-J}^J  e^{-y_j^2/2}e^{-iy_j\sqrt{2\beta H}}
\end{align}
and $y_j=j \delta y$.
\end{lemma}

The proof is in Appendix~\ref{app:unitapprox}.
Lemma~\ref{lem:discreteHTS} relates the exponential
operator with unitary operators that correspond
to evolutions under $\sqrt H$ for various times.
These evolution operators may not be available\,---\,in fact, computing the square
root of a Hamiltonian can be related to other computationally hard problems. 
To overcome this issue, we may construct a Hamiltonian $ H'$, acting on $m+m'_1$ qubits,
that satisfies
\begin{align}
\label{eq:H'def}
    ( H')^2 \ket \phi_m \ket 0_{m'_1} = (H \ket \phi_m) \ket 0_{m'_1} \;,
\end{align}
for all pure states $\ket \phi \in \cH_m$. Equation~\eqref{eq:H'def}
resembles the spectral gap amplification technique discussed in Ref.~\cite{SB13}. We discuss how to build $H'$ in Sec.~\ref{sec:SGA} for $H$ that is given as a linear combination of projectors.

Let $W_t:=e^{-itH'}$ be the evolution operator under $H'$.
Assume that there exists a unitary $W'_t$, which is an approximation of $W_t$ that acts on $m+m' > m+m'_1$ qubits and satisfies
\begin{align}
\label{eq:W'tcondition}
| \Tr[\bra 0_{m'_1} W_t \ket 0_{m'_1}] - \Tr[\bra 0_{m'} W'_t \ket 0_{m'}]| \le \eabs/8 \;,
\end{align}
for all $t$.
If $t_j := j \delta y \sqrt{2 \beta}$,
in Appendix~\ref{app:unitapprox} we show
\begin{align}
 \label{eq:tracerelation}
| \Tr [\bra 0 _{m'} X'_J \ket 0_{m'}] - \Tr[X_J]  | \le \eabs /4\;,
\end{align}
with
\begin{align}
\label{eq:discreteHTS2}
    X'_J:=\frac{\delta y}{\sqrt{2\pi}}\sum_{j=-J}^J  e^{-y_j^2/2} W'_{t_j} \;.
\end{align}
Our second approach
solves the PFP by using the relation
\begin{align}
\nonumber
 \Tr[\bra 0 _{m'} & X'_J \ket{0}_{m'}] = \\
   \label{eq:secondappx}
  &  \frac{\delta y}{\sqrt{2\pi}} (M+ 2 \sum_{j=1}^J  e^{-y_j^2/2} \Tr[\bra 0 _{m'} W'_{t_j}\ket 0 _{m'} ]\;,
\end{align}
which is an $\eabs/2$ approximation to $\cZ$.
We can use the construction of Fig.~\ref{fig:DQCkb_Alg} in the trace-estimation
algorithm of Fig.~\ref{fig:DQC1_Alg} to obtain $\Re\Tr [\bra 0_{m'} W'_{t_j}\ket 0_{m'}]$.

\section{Block-encoding and Hamiltonian simulation}
\label{sec:Uimp}

Our algorithms for estimating the partition function require implementing the unitary $W_H$,
which satisfies Eq.~\eqref{eq:iterate}, or simulating time evolution with a Hamiltonian $H'$, which satisfies Eq.~\eqref{eq:H'def}. 
These operations can be implemented efficiently for suitable specifications of $H$. 
Of interest in this paper are the cases where $H$ is given as a linear combination of unitary operators or projectors. These cases include physically relevant Hamiltonians that can be decomposed into a sum of tensor products of Pauli matrices, e.g., $k$-local Hamiltonians, Hamiltonians appearing in fermionic systems, and more. They also include the so-called frustration-free Hamiltonians
that are of relevance in quantum computing and condensed matter (cf., ~\cite{BT09,SB13,BS12}).

Our method may also be applied more broadly, e.g., to sparse Hamiltonians, but the resulting complexities may be large. This is because known methods for simulating sparse Hamiltonians may require, in general, a number of ancillary qubits $m'=\text{poly}(m)$, and the resulting complexities of
our method are exponential in $m'$.

\subsection{Implementing $W_H$} 
\label{sec:blockencoding}
For Algorithm 1.A, we will focus on the case where the Hamiltonian is specified as a linear combination of unitary operators.
Here, $H =\sum_{l=1}^{L} \alpha_l H_l $, where $\alpha_l >0$ and each $H_l$ is unitary.
We further assume that there exist quantum circuits, of maximum gate complexity $C_H$, that implement the $H_l$'s. 
Algorithm 1.A requires that $\|H\| \le 1$. We can satisfy this condition if we work with the renormalized Hamiltonian instead, as discussed in Sec.~\ref{sec:PFP}.
The renormalization is achieved by a simple pre-processing step, whose complexity is not significant. 
Therefore, with no loss of generality, we assume that $\sum_{l=1}^L\alpha_l=1$ and $L=2^{m'_1}$.

The unitary $W_H$ can be constructed following a procedure in Ref.~\cite{LC19}.
We define the unitary $G$ via
\begin{align}
   \label{eq:stateG}
   G \ket{0}_{m_1'} &:= \sum_{l=1}^{L} \sqrt{\alpha_l}\ket{l} \;,
\end{align}
which can be implemented with $O(L)$ two qubit gates, 
and the unitary
\begin{align}
\label{eq:U}
    U := \sum_{l=1}^{L} H_l \otimes \ketbra l _{m'_1} \;,
\end{align}
which acts on $m+m'_1$ qubits~\cite{BCC+14,CKS17}.
Let $U':=   U \otimes \ketbra 0^{(a)} +  U^\dagger \otimes \ketbra 1^{(a)}$,
where we added an ancilla qubit $a$, and $U'$ is a unitary acting on $m+m'$ qubits, with $m'=m'_1+1$. 
Defining $\tilde G = G\otimes {\text H}^{(a)}$, where ${\text H}^{(a)}$ denotes the Hadamard gate acting on  $a$, we obtain
\begin{align}
    W_H = (\one_m  \otimes (2 \tilde G \ketbra{0}_{m'} \tilde G^\dagger - \one_{m'}))\, \sigma_x^{(a)}\, U' \;.
\end{align}
Here, $\sigma_x^{(a)}$ corresponds to the qubit flip Pauli operator of the ancilla $a$
and $\one_n$ is the identity operator acting on $\cH_n$. 
In Ref.~\cite{LC19} it is shown that this choice of $W_H$ and $\tilde G$ satisfy Eq.~\eqref{eq:iterate} and, in particular, 
$\bra 0_{m'} \tilde G^\dagger U' \tilde G \ket 0 _{m'}=H$, which is a block encoding.
The gate complexity of $W_H$ is  $O(LC_H)$.

\subsection{Implementing $W'_t$} 
\label{sec:SGA}
Due to the requirement $H \ge 0$, here we assume that the Hamiltonian is specified as
$H= \sum_{l=1}^{L} \alpha_l H_l$, where $\alpha_l > 0$ and each $H_l$ is Hermitian and a projector, i.e.,
it satisfies $(H_l)^2 = H_l \ge 0$. That is, the eigenvalues of $H_l$ are 0,1. Note that the
case in Sec.~\ref{sec:blockencoding} can be reduced to this one if the eigenvalues of the unitaries are $\pm 1$ simply by shifting each unitary.
Algorithm 1.B also requires that $\|H\| \le 1$, which is satisfied if we work with the
renormalized Hamiltonian as before. With no loss of generality, we assume that $L+1 = 2^{m'_1}$, and $\sum_{l=1}^{L}\alpha_l=1$.

We now construct the operator $H'$ defined by Eq.~\eqref{eq:H'def} and further decompose it into a linear combination of a constant number of unitaries, improving upon a similar construction discussed in Ref.~\cite{CS16} that required $O(L)$ unitaries. This helps reduce the cost of trace-estimation in the one clean qubit model since fewer unitaries in the decompositions means that fewer ancilla qubits are required to implement $H'$. 

Following the technique of spectral gap amplification~\cite{SB13}, we obtain
\begin{align}
\label{eq:H'def2}
    H' = \sum_{l=1}^{L} \sqrt{\alpha_l} H_l \otimes \left(\ketbra{l}{0}_{m'_1}+\ketbra{0}{l}_{m'_1}\right) \;,
\end{align}
which acts on a space of $m +m'_1$ qubits. It requires computing the coefficients  $\sqrt{\alpha_l}$ in a simple pre-processing step.
Our goal is to simulate $e^{-itH'}$ and, to this end, we seek a decomposition of $H'$ as a linear combination 
of unitaries. Following Sec.~\ref{sec:blockencoding}, 
we define
\begin{align}
  G\ket 0_{m_1'}  & = \sum_{l=1}^{L} \sqrt{\alpha_l} \ket l \;.
\end{align}
We also define an operator which acts on $m + m'_1$ qubits, where ($H_0:= \one_m$)
\begin{align}
    X:= \sum_{l=0}^{L}  H_l \otimes \ketbra l_{m'_1} \;.
\end{align}
Then,
\begin{align}
H' = X G\ketbra{0}{0}_{m'_1} + h.c. \; .
\end{align}
Since $2\ketbra 0_{m'_1} = \one_{m'_1} -e^{i \pi \ketbra 0_{m'_1}}$ and 
$X = (\one+U)/2$ for unitary $U$, we can easily decompose $H'$ as a linear combination of, at most, 8 unitaries. If $C_H$ is an upper bound on the gate complexity of the unitaries $(2H_l - \one_m)$, the gate complexity of $U$ is $O(LC_H)$. Furthermore, the gate complexity of each of the unitaries in the decomposition of $H'$ will also be $O(LC_H)$, and the triangle inequality implies $\|H'\|=O(1) $.

Once the decomposition of $H'$ into a linear combination of unitaries is obtained, we can use the result in Ref.~\cite{LC19} for Hamiltonian simulation. This provides a quantum algorithm to implement a unitary $W'_t$, acting on $m+m'>m+m'_1$ qubits, such that
it approximates $W_t$ within additive error $\eabs/(8M)$ in the spectral norm. 
More specifically, the results in Ref.~\cite{LC19} imply
\begin{align}
\label{eq:W'tdef}
   | \Tr [\bra 0_{m_1'} W_t \ket 0_{m_1'}] - \Tr [\bra 0_{m'}  W'_t \ket 0_{m'}] | \le \eabs/8 \;,
\end{align}
which is the desired condition of Eq.~\eqref{eq:W'tcondition}.
The gate complexity for implementing $W'_t$  is $O ( L C_H ( |t| + \log(M/\eabs) ))$ and
the number of ancillary qubits is $m'=m'_1+m'_2$, with $m'_2=O(1)$ in this case.

\section{Algorithms}
\label{sec:alg}

We provide two algorithms in the one clean qubit model
for solving the PFP, based on the previous approximations.
We first focus on the case where the approximation
of the partition function is obtained within a given additive error
and next use this case 
to obtain an approximation within given relative error.
To this end, we assume that we know $\cZ_{\max}$
such that $\cZ_{\max} \ge \cZ > 0$. In particular,
under our assumptions, we may choose
$\cZ_{\max} = M e^{\beta}$ or $\cZ_{\max} = M$, depending on whether 
we work with the first or second approximation, respectively.

\subsection{Estimation within additive error}

Algorithm 1.A provides an estimate according to Eq.~\eqref{eq:firstappx}. Specifically, for given $\eabs>0$ and $\delta>0$, it outputs $\hat \cZ$
such that $|\hat \cZ-\cZ| \le \epsilon_{\rm abs}$ with probability at least $(1-\delta)$. 
The algorithm itself only assumes $\|H\| \le 1$ and the existence of a procedure $W_H$ that acts on $m+m'$ qubits and satisfies Eq.~\eqref{eq:iterate}.
In Sec.~\ref{sec:blockencoding} we discussed the construction of $W_H$ for Hamiltonians that are given as a linear combination of $L$ unitaries with $m'=O(\log L)$.
This is required to build unitaries $V(k)$ acting on $m+2m'$ qubits according to Fig.~\ref{fig:DQCkb_Alg}: $V(k)$ is obtained by replacing $W$ with $\tilde G^\dagger(W_H)^k\tilde G$. 

\begin{algorithm}[H]
\renewcommand{\thealgorithm}{}
\floatname{algorithm}{\hspace{3.3cm} Algorithm 1.A}
\setstretch{1.25}
\vspace{1mm}
{\bf Input:} $\epsilon_{\rm abs}>0$, $\delta>0$, $\cZ_{\max} > 0$. \\
-- Obtain $K$ according to Lemma~\ref{lem:finitecheby}.\\
-- Set $\varepsilon = \eabs/(2e^\beta)$, $\delta_0= \delta/K$, and obtain $Q$ \\
\phantom{--} according to Cor.~\ref{cor:Hoef2}. \\
-- For each $k=1,\ldots,K$: \\
\phantom{xxxx} Run the trace estimation algorithm $Q$ times with \\
\phantom{xxxx} unitary $V(k)$. Obtain $\hat \chi_k =2^{m+m'}s_x$ where $s_x$ is \\
\phantom{xxxx}
 the average of the measurement outcomes of $\sigma_x$. \\
-- Compute $\hat \cY= I_0(\beta)M + 2 \sum_{k=1}^K (-1)^kI_k(\beta) \hat \chi_k$.\\
{\bf Output:} $\hat \cZ =\hat \cY$ if $\cZ_{\max}>\hat \cY \ge 0$, $\hat \cZ=0$ if $\hat \cY <0$,\\ 
 \phantom{\bf Output:} and $\hat \cZ=\cZ_{\max}$ otherwise.\\
\vspace{-3mm}
 \caption{}
\label{alg:cheby}
\end{algorithm}

Algorithm 1.B provides an estimate according to Eq.~\eqref{eq:secondappx}. It assumes $\|H\| \le 1$, $H \ge 0$, and the existence
of a procedure $W'_t$ that approximates the evolution under a Hamiltonian $H'$ which satisfies Eq.~\eqref{eq:H'def}. 
We described an implementation of $W'_t$ for Hamiltonians that are given as a linear combination of projectors in Sec.~\ref{sec:SGA}.
This procedure is required to build
unitaries $V'(t)$ acting on $m+2m'$ qubit  according to Fig.~\ref{fig:DQCkb_Alg}: $V'(t)$ is obtained
by replacing $W$ with  $W'_t$. In the following, $t_j:= y_j\sqrt{2\beta}$ with $y_j = j \delta y$.
\begin{algorithm}[H]
\renewcommand{\thealgorithm}{}
\floatname{algorithm}{\hspace{3.3cm} Algorithm 1.B}
\setstretch{1.25}
\vspace{1mm}
{\bf Input:} $\epsilon_{\rm abs}>0$, $\delta >0$, $\cZ_{\max} > 0$. \\
-- Obtain $J$ and $\delta y$ according to Lemma~\ref{lem:discreteHTS}.\\
-- Set $\varepsilon = \eabs/4$, $\delta_0= \delta/J$, and obtain $Q$ \\
\phantom{--} according to Cor.~\ref{cor:Hoef2}. \\
-- For each $j=1,\ldots,J$: \\
\phantom{xxxx} Run the trace estimation algorithm $Q$ times with \\
\phantom{xxxx} unitary $V'(t_j)$. Obtain $\hat \chi_j =
2^{m+m'} s_x$ where $s_x$ is \\
\phantom{xxxx} the average of the measurement outcomes of $\sigma_x$.
 \\
-- Compute $\hat \cY=(\delta y/\sqrt{2\pi})(M+ 2 \sum_{j=1}^J e^{-y_j^2/2} \hat \chi_j)$.\\
{\bf Output:} $\hat \cZ =\hat \cY$ if $\cZ_{\max}>\hat \cY \ge0$, $\hat \cZ=0$ if $\hat \cY<0$,\\ 
 \phantom{\bf Output:} and $\hat \cZ=\cZ_{\max}$ otherwise.\\
\vspace{-3mm}
 \caption{}
\label{alg:HST}
\end{algorithm}

For simplicity,
in the following we will refer to both Algorithms 1.A and 1.B as ${\textsc{ EstimatePF}}(\eabs,\delta,\cZ_{\max})$.

\subsection{Estimation within relative error}
\label{sec:relative_error}
The PFP is formulated in terms of a relative error $\erel$ and error probability $\delta$. 
Given a known lower bound $\cZ_{\text{min}}$ on $\cZ$, a naive approach to achieve the desired relative precision would be to simply perform an additive-error estimation with target precision $\eabs'=\erel\cZ_{\text{min}}$. The drawback of this is that the resulting complexity will scale as $1/(\erel\cZ_{\text{min}})^2$, resulting in significant overhead if $\cZ\gg\cZ_{\text{min}}$.

We therefore develop a classical procedure (Algorithm 3 in Appendix~\ref{appx:est_rel_err}) for obtaining an estimate within a desired relative precision from multiple additive estimations that avoids the overhead mentioned above. In this section, we discuss its application to the PFP. The procedure, however, goes beyond the PFP and may be of independent interest. Note that this approach is distinct from other known methods for obtaining relative approximations, which involve either expressing the measurable quantity as a telescopic product of ratios~\cite{DFK91,JS93,PW09}, or estimating the logarithm of the quantity (i.e. the free energy in our case) within additive error~\cite{Bar16,HMS19}. 
 
\begin{algorithm}[H]
\renewcommand{\thealgorithm}{}
\floatname{algorithm}{\hspace{3.3cm}Algorithm 2}
\setstretch{1.25}
\vspace{1mm}
{\bf Input:} $\epsilon_{\rm rel}>0$, $\delta >0$, $\cZ_{\max}> 0$. \\
-- Set $r=0$, $\cZ_0=\cZ_{\max}$, and $\hat \cZ_0 = 0$.\\
-- While $\cZ_r > \hat \cZ_r$:\\
\phantom{xxxx}$r \leftarrow r+1$. \\
\phantom{xxxx}Set $\cZ_r =\cZ_{\max} /2^{r}$, $\eabs(r)=\erel\cZ_r/2$, and \\
\phantom{xxxx}$\delta'(r) = \tfrac{6}{\pi^2}\left(\delta/r^2\right)$.\\
\phantom{xxxx}$\hat \cZ_r = {\textsc{EstimatePF}}(\eabs(r),\delta'(r), \cZ_{\max})$\\
{\bf Output:} $\hat \cZ=\hat \cZ_r$.\\
\vspace{-3mm}
 \caption{}
\end{algorithm}

The value of $r$ at which Algorithm 2 stops is a random variable $R$.
For any given instance of the PFP, we show that
the expected value of $R$, $E_R$, is bounded as $E_R\le(\lceil \log_2(\cZ_{\max}/\cZ)\rceil + 3)$, and
the probability of $E_R$ going past this value decays super-exponentially with $(E_R-R)$\,---\,see Appendix~\ref{appx:est_rel_err}. This allows us to bound the expected complexity of Algorithm~2 in Appendix~\ref{app:rel_err_complexity} and obtain a significant improvement over the naive approach.

\section{Correctness}
\label{sec:CC}

Algorithm 1.A obtains $K$ estimates of $\Tr[\bra 0_{m'}\tilde G^\dagger (W_H)^k \tilde G \ket 0_{m'}]$, each within additive error
$\varepsilon=\eabs/(2e^\beta)$ and probability at least $(1-\delta_0)$. It follows that, with probability at least $(1-\delta_0)^K \ge (1-\delta)$,
\begin{align}
\nonumber
    \ | & \hat \cY - \Tr[S_K]|    \\ 
    & \le 2  \sum_{k=1}^K I_k(\beta) \left|\hat \chi_k -\Tr [\bra 0_{m'}\tilde G^\dagger (W_H)^k \tilde G \ket 0_{m'}] \right| \\
    & \le \eabs/2 \;,
\end{align}
where we used $\sum_{k=1}^K I_k(\beta) \le e^\beta/2$.
In addition, Lemma~\ref{lem:finitecheby} 
implies $ |\Tr[S_K] - \cZ| \le \eabs/2$ 
and thus $|\hat \cY - \cZ| \le \eabs$ with probability at least  $(1-\delta)$.
We can then choose $\hat \cY$ as the estimate $\hat \cZ$ for the partition function in all cases.
However, if $\hat \cY>\cZ_{\max}$ or $\hat \cY<0$, we can set $\hat \cZ=\cZ_{\max}$ or $\hat \cZ=0$ respectively, and still satisfy
\begin{align}
\label{eq:Alg1correct}
    |\hat \cZ - \cZ| \le \eabs \;,
\end{align}
with probability at least $1-\delta$.  

Algorithm 1.B obtains $J$ estimates of 
$\Tr [\bra 0 _{m'} W'_{t_j}\ket0_{m'}]$, each within additive error 
$\varepsilon=\eabs/4$ and error probability $\delta_0$. It follows that, with probability at least $(1-\delta_0)^J \ge (1-\delta)$,
\begin{align}
\nonumber
    \ | & \hat \cY - \Tr[\bra{0}_{m'}X'_J \ket 0_{m'}]|    \\ 
    & \le 2 \delta y/(\sqrt{2 \pi})  \sum_{j=1}^J e^{-y_j^2/2} \left|\hat \chi_j -\Tr [\bra 0_{m'} W'_{t_j} \ket 0_{m'}] \right| \\
    & \le \eabs/2 \;,
\end{align}
where we used $ 2 \delta y/(\sqrt{2 \pi})  \sum_{j=1}^J e^{-y_j^2/2} \le 2$.
In addition, Appendix~\ref{app:unitapprox} and Lemma~\ref{lem:discreteHTS} imply $|\Tr[\bra{0}_{m'}X'_J \ket 0_{m'}]-\Tr[X_J]| \le \eabs/4$ and $|\Tr[X_J]-\cZ| \le \eabs/4$, so that 
$|\Tr[\bra{0}_{m'}X'_J \ket 0_{m'}] -\cZ| \le \eabs/2$.
Thus $|\hat \cY - \cZ| \le \eabs$ and Eq.~\eqref{eq:Alg1correct} is satisfied
with probability at least $1-\delta$.

Finally, the proof of the correctness of Algorithm 2 follows directly from the analysis in Appendix \ref{appx:est_rel_err}. 
The main observations are that the final $\eabs(r)$ is sufficient for the desired relative precision
and $\Pi_{r\ge 1} (1-\delta'(r)) \ge (1-\delta)$.
This algorithm returns an estimate $\hat \cZ$ that satisfies
\begin{align}
      | \hat \cZ -  \cZ   | \le \erel  \cZ\;,
\end{align}
with probability at least $(1-\delta)$.


\section{Complexity}
\label{sec:complexity}

\subsection{Additive error} 
The complexity of the algorithms can be determined
from the total number of uses of the trace-estimation algorithm 
and the complexity of each use. 
The operation $V(k)$ in Algorithm 1.A uses $W_H$ at most $K=O(m+\beta + \log_2(1/\eabs))$ times and the gate complexity of $W_H$ is $O(LC_H)$ as discussed in Sec.~\ref{sec:blockencoding}.
The gate complexity of $V(k)$ is then
\begin{align}
\label{eq:encodingLCU}
    C_V = O \left( LC_H (m+\beta+\log(1/\eabs))  \right) \;.
\end{align}
Algorithm 1.B requires computing $\Tr [\bra 0_{m'} W'_t \ket 0_{m'}]$ within additive precision $\Theta(\eabs)$. 
The largest time is $t_J := J \delta y \sqrt{2\beta}=O (\sqrt {\beta(m + \log_2(1/\eabs))}\,)$. 
The gate complexity of $V'(t)$ is determined by that of $W'_t$ 
and is
\begin{align}
\label{eq:encodingFF}
    C_{V'}= O \Big(& L C_H  \sqrt {m+ \log(1/\eabs)}\nonumber \\ 
    &\times\left(\sqrt\beta + \sqrt {m+ \log(1/\eabs)}\right) \Big) \;.
\end{align}

For given $\eabs$ and $\delta$, Algorithm 1.A uses the trace estimation algorithm $Q.K$ times, while
Algorithm 1.B uses it $Q.J$ times, for proper choices of $Q$, $K$, and $J$ given by Lemmas~\ref{lem:Hoef},~\ref{lem:finitecheby} and~\ref{lem:discreteHTS} respectively. Moreoever, the number of ancilla qubits needed to perform the trace estimations in either case is $O(\log_2 L)$. 
Using the results of Secs.~\ref{sec:DQC1} and~\ref{sec:LCU} and Eqs.~\eqref{eq:encodingLCU} and~\eqref{eq:encodingFF}, the respective asymptotic complexities are given by
\begin{align}
\label{eq:1Acomp}
    T_{1A} &= \tilde O \Bigg(\frac{e^{2\beta} M^2}{\eabs^2}\log(1/\delta)L^3C_H(\beta^2+m^2)\Bigg) \;, \\
\label{eq:1Bcomp}
    T_{1B} &= \tilde O \Bigg(\frac{ M^2}{\eabs^2} \log(1/\delta){L^3 C_H }m(\beta+m)\Bigg) \;,
\end{align}
where the $\tilde O$ notation hides factors that are polylogarithmic in $1/\eabs$, $\beta$ and $m$.

Note that the assumptions on $H$ in Sections~\ref{sec:blockencoding} and~\ref{sec:SGA} are closely related. Indeed, the two cases can be connected via simple transformations, such as adding or subtracting an overall constant to and scaling the Hamiltonian. These transformations can change $\cZ$ by an overall constant factor that may be exponentially large or small in $\beta$, potentially incurring in complexity overheads.
Hence the main deciding factor for choosing between Algorithms 1.A and 1.B is the specification of the Hamiltonian, which affects the complexity directly via $C_V$ and $C_{V'}$ in Eqs.~\eqref{eq:encodingLCU} and~\eqref{eq:encodingFF} respectively. The specification of $H$ also determines the $L1$-norm of the coefficients in the decomposition of $H$, which affects the complexities through their dependence on $\beta$. This is because we renormalize $H$ by the $L1$-norm at the outset\,---\,see Sec.~\ref{sec:PFP}.

\subsection{Relative error}
Algorithm 2 calls $\textsc{EstimatePF}$ with variable absolute error $\eabs(r)$ and error probability $\delta'(r)$ for $r=1,...,R$, and the complexity of each call can be obtained from Eqs.~\eqref{eq:1Acomp} and~\eqref{eq:1Bcomp}).
The number of times Algorithm 2 uses $\textsc{EstimatePF}$ is a random variable and hence the complexity of any one instance of Algorithm 2 is a random number. However, the probability of requiring more than the expected number of uses of $\textsc{EstimatePF}$ decays super-exponentially. This allows us to obtain bounds on the expected complexity of Algorithm 2 as:
\begin{align}
    T_{2A} &= \tilde O \left( \frac{e^{2\beta} M^2 }{\erel^2 \cZ^2} {\log(1/\delta)L^3 C_H  }\beta^3 (\beta^2+m^2)  \right) \;, \label{eq:Alg2Acost}\\
    T_{2B} &= \tilde O \left( \frac{ M^2 }{\erel^2 \cZ^2} {\log(1/\delta)L^3 C_H  }  m(\beta+m)  \right) \;, \label{eq:Alg2Bcost}
\end{align}
where we dropped terms that are polylogarithmic in $1/\erel$, $\beta$, $m$ and $M/\cZ$\,---\,see Appendix~\ref{app:rel_err_complexity} for details.

The dominating factor in $T_{2A}$ is $O((M e^\beta/(\erel \cZ ))^2L^3)$ and in $T_{2B}$ is $O((M /(\erel \cZ ))^2L^3)$. While the latter appears to be exponentially smaller in $\beta$, the partition function may also be exponentially smaller under the condition $H \ge 0$ in this case.


\section{DQC1 completeness}
\label{sec:complete}

In Ref.~\cite{Bra08} it was shown that, under certain conditions, estimating 
the partition function to within additive error is DQC1-hard, i.e., any problem in DQC1 can be efficiently reduced  to that of estimating the partition function.
We now show that such a version of the partition function problem is in fact DQC1-complete\,---\,our method provides a polynomial-time algorithm in the one clean qubit model.
In the following, $\lambda_{\min}(A)$ refers to the lowest eigenvalue of a Hermitian operator $A$.

\begin{definition}[PFP-additive~\cite{Bra08}]\label{def:PFP_add}
    We are given a Hamiltonian $\tilde H$ acting on $m$ qubits, $\tilde H=\sum_{l=1}^{\tilde L} \tilde h_l$, and three real numbers $\tilde \beta >0$, $\delta<1$, and $\epsilon >0$.
   Each $\tilde h_l$ acts on at most $k$ qubits and has bounded operator norm $\|\tilde h_l\| =O( \poly(m))$.
   We are also given a lower bound $\lambda$ to the ground state energy of $\tilde H$, i.e. $\lambda \le \lambda_{\min}(\tilde H)$.
   The goal is to find a number $\hat \cY$ such that, with probability at least $(1-\delta)$,
    \begin{align}
    \label{eq:PFPadditive}
        \left|\hat \cY -\frac{\tilde \cZ}{2^m e^{-\tilde \beta \lambda}} \right| \leq \epsilon \;,
    \end{align}
    where $\tilde \cZ = \Tr[e^{-\tilde \beta \tilde H}]$.
\end{definition}
Our result in this section is:
\begin{theorem}
\label{thm:PFP-additive}
    PFP-additive is \dqc1-complete for $\tilde L=O(\poly(m))$, $\tilde \beta =O( \poly(m))$,  $\delta$ a constant such that  $0<\delta<1/2$, $\epsilon=\Omega(1/\poly(m))$, $k=O(\log(m))$, and $\lambda = \sum_{l=1}^{\tilde L}\lambda_{\min}(\tilde h_l)$.
\end{theorem}
\begin{proof}
     
    Reference~\cite{Bra08} shows that PFP-additive is DQC1-hard under the stated conditions. 
    It remains to be shown that this problem is indeed in the DQC1 complexity class; we will prove this using Algorithm 1.B.
    
    Let $ H := (\tilde H - \lambda)/\gamma$, where $\gamma:=2 \sum_{l=1}^{\tilde L} \|\tilde h_l\|$, and $ \beta: = \tilde \beta \gamma$ such that $ H \ge 0$ and $\| H \| \le 1$.
    Moreover, the estimate $\hat \cY$ in Eq.~\eqref{eq:PFPadditive} is equivalent to an estimate of $\cZ$ within additive error $\eabs=\epsilon M$, since $ \cZ = \tilde \cZ e^{\tilde \beta \lambda}$. Given access to evolutions under a Hamiltonian $H'$ that satisfies Eq.~\eqref{eq:H'def},
    we can obtain $\hat \cY$. 
        
    We now show how to construct $H'$ and approximate its evolution operator
    in  time polynomial in $\tilde L$ and $2^k$; that is, 
    time polynomial in $m$ under the assumptions. 
    First, we classically obtain a matrix representation for each $\tilde h_l$, whose dimensions are, at most,
    $2^k \times 2^k$.
    We also compute $\lambda_{\min}(\tilde h_l)$ and $\|\tilde h_l\|$ for each $l$, and obtain $\gamma$. 
    Next, we construct the operators or matrices $h_l = (\tilde h_l - \lambda_{\min}(\tilde h_l))/\gamma$ and note that
    $H = \sum_{l=1}^{\tilde L} h_l$, with $h_l \ge 0$. We proceed with the spectral decomposition of each $h_l$
    and write it as $h_l = \sum_{d=1}^{2^k} \bar{\alpha}_l^d H_l^d $, where $\bar{\alpha}_l^d\ge0$, each $H_l^d$ satisfies $(H_l^d)^2=H_l^d$ and acts on, at most, $k$ qubits.
    
    Using this decomposition and consolidating indices, we can write $H = \sum_{l'=1}^{L} \alpha_{l'} H_{l'}$, where $\alpha_{l'} \ge 0$ and the $H_{l'}$ are projectors acting on at most $k$ qubits. It follows that each unitary $2 H_{l'} - \one_m$ can be implemented using $O(\poly(2^k))$ quantum gates, i.e., has complexity $C_H=O(\poly(m))$.
    We also note that $L=2^k \tilde L=O(\poly(m)) $ and $\alpha=\sum_{l'=1}^L \alpha_{l'}= O(\poly(m))$. Lastly, we renormalize $H\gets H/\alpha$ and $\beta\gets \beta\alpha$ in order to fit the framework of Algorithm 1.B. The complexity of all classical steps is $O(\poly(m))$, and the inverse temperature after renormalizing $H$ (i.e.,  $\beta\gets \beta\alpha$) is $O(\poly(m))$. 
    
    We may now use Algorithm 1.B to return an estimate of $\cZ$ within additive error $\epsilon M$, which solves PFP-additive. The complexity analysis is essentially identical to that used for computing $T_{1B}$ in Equation~\eqref{eq:1Bcomp}. The overall gate complexity of Algorithm 1.B to obtain $\hat \cY$  turns out to be $\tilde O(\beta m^2 L^3 C_H/\epsilon^2)$, which is $O(\poly(m))$ under the stated conditions.

    In summary, we showed that the complexity of all steps to obtain $\hat \cY$
    is polynomial in $m$ and thus PFP-additive is in the DQC1 complexity class.
    
     \end{proof}


\section{Discussions}
\label{sec:conc}

We provided a method to compute partition functions of quantum systems in the one clean qubit model. 
For a given relative precision and probability of error, and when the Hamiltonian is positive semi-definite, the complexity of our method is almost linear in $(M/\cZ)^2$.
Our algorithm can outperform classical methods whose complexity is polynomial in $M$ whenever $\cZ$ is sufficiently large. However, in the general case, our method may be inefficient with complexity that scales polynomially in $M$, which is expected due to the hardness of estimating partition functions.
As for any algorithm that computes $\cZ$, this can be a drawback in an implementation when $M \gg 1$.
 
For our result, we developed a classical algorithm for attaining an estimate of a quantity within desired relative precision based on a sequence of approximations within predetermined additive errors. This result is applicable in a fairly general setting, and could be of independent interest, e.g., in other DQC1 algorithms \cite{CM18,subramanian2019quantum}.

We also showed that, under certain constraints on the inverse temperature and Hamiltonian, the problem
of estimating partition functions within additive error is complete for the DQC1 complexity class.
This result suggests that no efficient classical algorithm for this problem exists,
while our method is efficient for those instances. It also demonstrates the power 
of the one clean qubit model for solving problems of relevance in science.

Several simple variants of our method may be considered. For example, instead of estimating each of the traces appearing in Eqs.~\eqref{eq:firstappx} and~\eqref{eq:secondappx} and computing the linear combination, we could sample each unitary with a probability that is proportional to the corresponding coefficient. We could also aim at improving our error bounds by avoiding the union bound and noticing that $\cZ$ is ultimately obtained by sampling independent $\pm 1$ random variables. However, these improvements may not reduce the complexity significantly. The reason is that we use fairly efficient approximations to $e^{-\beta H}$, where the number of terms in the corresponding linear combinations have a mild dependence on the relevant parameters of the problem.

Quantum algorithms for partition functions that have complexity almost linear in $\sqrt{M/\cZ}$ exist (cf.~\cite{PW09,CS16}), but but are formulated in the standard circuit model of quantum computation and thus require a number of pure qubits that is $\Theta(m)$.
Nevertheless, achieving a scaling that is almost linear in $\sqrt{M/\cZ}$ in the one clean qubit model  would imply a quadratic quantum speedup for unstructured search under the presence of oracles. Such a speedup is ruled out from a theorem in Ref.~\cite{KL98}.

\section{Acknowledgements}
We thank the anonymous reviewers whose comments have helped improve this manuscript.
ANC thanks David Poulin for helpful discussions. ANC acknowledges the Center for Quantum Information and Control, University of New Mexico and the Theoretical Division, Los Alamos National Laboratory where a part of this work was done. This work was supported by the Laboratory Directed Research and Development program of Los Alamos National Laboratory and by the U.S. Department of Energy, Office of Science, Office of Advanced Scientific Computing Research, Quantum Algorithms Teams and Accelerated Research in Quantum Computing programs. Los Alamos National Laboratory is managed by Triad National Security, LLC, for the National Nuclear Security Administration of the U.S.
Department of Energy under Contract No. 89233218CNA000001.

\appendix

\section{Proof of Lemma~\ref{lem:Hoef}}
\label{app:prooflemma1}

Let $X_1,\ldots,X_{Q}$ be a set of independent and identically distributed 
random variables and $X_i \in \{1,-1\}$. These variables
can be associated with the outcomes of the $Q$ projective measurements of $\sigma_x$.
Let $s_x:=\sum_i X_i/Q$. According to Hoeffding's inequality~\cite{Hoe63}, 
\begin{align}
\label{eq:Hoef}
    \Pr(|s_x- \langle \sigma_x \rangle| \ge t) \le 2 \exp \left(- \frac { t^2 Q}{2} \right) \;,
\end{align}
where $\langle \sigma_x \rangle$ is also the expected value of each $X_i$. 
Our estimate is the random variable $\hat \xi_V = 2^n s_x$.
 We can use Eq.~\eqref{eq:Hoef} to obtain
\begin{align}
   & \Pr(| \hat \xi_V - 2^n\langle \sigma_x  \rangle| \ge \varepsilon) \le 2 \exp \left(- \frac {\varepsilon^2 Q} {2^{2n+1}} \right) \;.
\end{align}
Then, it suffices to choose 
\begin{align}
    Q = \left\lceil \frac {2^{2n+1}} {\varepsilon^2} \log\left( \frac 2 {\delta_0}\right) \right \rceil
\end{align}
to satisfy $ \Pr(| \hat \xi_V -2^n \langle \sigma_x \rangle| \ge \varepsilon) \le \delta_0$
or, equivalently,
\begin{align}
     \Pr(| \hat \xi_V - 2^n \langle \sigma_x \rangle| < \varepsilon) \ge (1-\delta_0) \;.
\end{align}
Last, we note that $2^n \langle \sigma_x \rangle = \Re \Tr[V]$.

\section{Approximation of the exponential operator in terms of Chebyshev polynomials}
\label{app:chebyapprox}

Let $I_k(z)$ be the modified Bessel function of the first kind.
The generating function is $ e^{-z \cos q }=  \sum_{k=-\infty}^\infty I_k(-z) e^{ikq}$,
which implies
\begin{align}
\label{eq:expcheb}
    e^{-\beta H} &
     =  \sum_{k=-\infty}^\infty (-1)^k I_k(\beta) T_k(H) 
    \; ,
\end{align}
where $T_k(x)$ is the $k$-th Chebyshev polynomial of the first kind.
Equation~\eqref{eq:expcheb} was obtained using  
$I_k(z)=I_{-k}(z)$, $T_k(x)=T_{-k}(x)=\cos(k \arccos{x})$, $I_k(-z)=(-1)^k I_k(z)$,
and $T_k(-x)=(-1)^k T_k(x)$. 

\subsection{Proof of Lemma~\ref{lem:finitecheby}}

We wish to approximate the exponential
operator by a finite sum of Chebyshev polynomials in $H$.
For $|x|\le 1$, $|T_k(x)| \le 1$
and thus $\|T_k(H)\|_1 \le M$. This implies
\begin{align}
\label{eq:appS_K}
   \| S_K - e^{-\beta H} \|_1 \le M \sum_{|k|>K}  |I_k(\beta)| \;,
\end{align}
where $S_K = \sum_{k=-K}^K (-1)^kI_k(\beta)T_k(H)$.

To bound the right hand side of Eq.~\eqref{eq:appS_K}, we note that for $k \ge 0$ the following holds
\begin{align}
    I_k(\beta)&  = \left( \frac \beta 2 \right)^ k \sum_{r=0}^\infty
    \frac{(\beta^2/4)^r}{r! (r+k)!} \\
    & \le \left( \frac \beta 2 \right)^ k \frac {I_0(\beta)} {k!}  \;,
\end{align}
where we used $(r+k)! \ge k! r!$.
Additionally, $k! >(k/e)^k$
and $I_0(\beta) = \int_0^\pi e^{\beta \cos \theta} d \theta /\pi \le e^\beta$.
It follows that
\begin{align}
   \| S_K - e^{-\beta H} \|_1 & \le M e^\beta 2 \sum_{k > K} \left(\frac {\beta e}{2K} \right)^k \;.
\end{align}
Assume $K \ge \beta e$ and $K \ge 1$ in general, and $\beta \ge 0$. Then,
\begin{align}
\label{eq:SKbound}
     \| S_K - e^{-\beta H} \|_1 & \le M e^\beta  \frac 1 {2^{K-1}} \;.
\end{align}
To bound the right hand side by $\eabs/2$ we choose
\begin{align}
    K = \left \lceil m + e \beta + \log_2 (1/\eabs) +2 \right \rceil \;.
\end{align}
It is easy to show that both assumptions, $K \ge \beta e$ and $K \ge 1$, are satisfied
with this choice.

\section{Approximation of the exponential operator as a linear combination of unitaries}
\label{app:unitapprox}
For $\beta \ge 0$, the Fourier transform of the Gaussian gives
\begin{align}
\label{eq:FTGaussian}
    e^{-\beta x^2} = \frac 1 {\sqrt{2\pi}} \int dy \; e^{-y^2/2} e^{-i y \sqrt{2\beta} x} \;.
\end{align}
This formula can be used to obtain the Hubbard-Stratonovich transformation: if $x^2=\lambda \ge 0$
corresponds to the eigenvalue of $H$, then we can
replace $x^2$ by $H$ in Eq.~\eqref{eq:FTGaussian} and obtain
Eq.~\eqref{eq:HST}.

\subsection{Proof of Lemma~\ref{lem:discreteHTS}}

The Poisson summation formula implies
\begin{align}
    \frac {\delta y}{\sqrt{2\pi}} \sum_{j=-\infty}^\infty e^{-y_j^2/2} e^{-iy_j\sqrt{2\beta \lambda}} = \sum_{k=-\infty}^\infty e^{-\omega_k^2/2} \; ,
\end{align}
where $\delta y > 0$, $y_j=j \delta y$, and $\omega_k=2\pi k/\delta y + \sqrt{2\beta \lambda}$.   
Let us choose $\delta y$ and $1 \ge \epsilon >0$ so that
\begin{align}
\label{eq:deltaycondition}
    (2\pi/\delta y) \ge \sqrt{2 \beta \lambda} + \sqrt{2\log(5/\epsilon)} \;.
\end{align}
Then, for $|k| \ge 1$, we have $\omega_k^2 \ge k^2 2 \log(5/\epsilon)$,
where we considered the worst case in which $\lambda=0$. We obtain
\begin{align}
    \sum_{k \ne 0} e^{-\omega_k^2/2} &\le 2  \sum_{k =1}^\infty e^{-k^2 \log(5/\epsilon)} \\
    & \le 2  \sum_{k =1}^\infty (\epsilon/5)^k \\
    & = 2 \frac{\epsilon/5}{1-\epsilon/5} \le \epsilon/2 \;.
\end{align}
It follows that
\begin{align}
    \left | \frac {\delta y}{\sqrt{2\pi}} \sum_{j=-\infty}^\infty e^{-y_j^2/2} e^{-iy_j\sqrt{2\beta \lambda}} - e^{-\beta \lambda}\right | \le \epsilon/2 \;.
\end{align}
Moreover, we can choose $J$ such that
\begin{align}
\label{eq:yJcondition}
    y_J \ge \sqrt{6 \log(2/\epsilon)} \ge 2 \;,
\end{align}
and
\begin{align}
    \frac {\delta y}{\sqrt{2\pi}} \sum_{|j|>J} e^{-y_j^2/2} &\le \int_{y_J}^\infty dy \; e^{-y^2/2} \\
    & \le \int_{y_J}^\infty dy \; e^{-y_J .y/2} \\
    & = (2/y_J) e^{-y_J^2/2} \\
    & \le (\epsilon/2)^3 \le \epsilon/2 \;.
\end{align}
In particular, we can choose $ \delta y = ( \sqrt{\beta} + \sqrt{ \log(5/\epsilon)})^{-1} $
and $J = \lceil 3  (\sqrt{\beta} + \sqrt{\log(5/\epsilon)}) \sqrt{\log(5/\epsilon)} \rceil$
so that
\begin{align}
\label{eq:HSTapprox3}
\left | \frac {\delta y}{\sqrt{2\pi}} \sum_{j=-J}^J e^{-y_j^2/2} e^{-iy_j\sqrt{2\beta \lambda}} - e^{-\beta \lambda}\right | \le \epsilon \;.
\end{align}
Note that we assumed $0 \le \lambda \le 1$.
Larger values of 
$J$ and/or a smaller values of $\delta y$
will also imply Eq.~\eqref{eq:HSTapprox3}.
Lemma~\ref{lem:discreteHTS} then follows
from replacing $\lambda$ by $H$ and $\epsilon$
by $\eabs/(4M)$ in Eq.~\eqref{eq:HSTapprox3}:
\begin{align}
\label{eq:HSTapprox4}
   \left \| \frac {\delta y}{\sqrt{2\pi}} \sum_{j=-J}^J e^{-y_j^2/2} e^{-iy_j\sqrt{2\beta H}} - e^{-\beta H} \right \|_1 & \le  M\epsilon \\
   \nonumber
   & \le \eabs/4 \;.
\end{align}
We can simplify the expressions for $\delta y$ and $J$
using $\log(20 M/\eabs) \le 4( m + \log_2(1/\eabs))$, where $m=\log_2 (M)\ge 1$.
In particular, we can choose 
\begin{align}
\delta y &= \left( 2 (\sqrt \beta + \sqrt{m+\log_2(1/\eabs)}) \right)^{-1} \; , \\
    J & =\left \lceil 12  (\sqrt{\beta} + \sqrt{m+ \log_2(1/\eabs)} ) \sqrt{m + \log_2(1/\eabs)} \right \rceil\;.
\end{align}

\subsection{Proof of Eq.~\eqref{eq:tracerelation}}

Let $\ket{\psi_\lambda}$ be an eigenvector of $H$ of eigenvalue $\lambda \ge 0$, that is
$H \ket{\psi_\lambda} = \lambda \ket{\psi_\lambda}$.
Then, if $y_j = j \delta y$ and $t_j = j \delta y \sqrt{2\beta}$, we can write $ (X_J \ket{\psi_\lambda}_m) \ket 0_{m'_1}$
using Eq.~\eqref{eq:discreteHTS1} as
\begin{align}
\label{eq:X'Jaction}
    \left( \frac {\delta y}{\sqrt {2 \pi} } \sum_{j=-J}^J e^{-y_j^2/2}
    e^{-i t_j\sqrt{ \lambda}} \ket{\psi_\lambda}_m \right) \ket 0_{m'_1}
   \;.
\end{align}
If $\lambda=0$, the Hamiltonian $H'$ of Eq.~\eqref{eq:H'def}
has $\ket{\psi_\lambda}_m \ket 0_{m'_1}$ as eigenvector of eigenvalue 0.
Otherwise, $H'$ leaves the subspace spanned by $\{\ket{\psi_\lambda}_m \ket 0_{m'_1} , H' \ket{\psi_\lambda}_m \ket 0_{m'_1}\}$ invariant. We let
$\sket{\psi_\lambda^\perp}$
be the normalized state in this subspace that is orthogonal to $\ket{\psi_\lambda}_m \ket 0_{m'_1}$.
The two-dimensional representation of $H'$ is
\begin{align}
    H'_\lambda = \begin{pmatrix} a_\lambda & b_\lambda \cr b_\lambda & c_\lambda \end{pmatrix} \; .
\end{align}
With no loss of generality, $a_\lambda,b_\lambda,c_\lambda \in \mathbb R$.
According to Eq.~\eqref{eq:H'def}, $H'_\lambda$ must satisfy
\begin{align}
    (H'_\lambda)^2=
     \begin{pmatrix} \lambda & 0 \cr 0 & \gamma \end{pmatrix} \; ,
\end{align}
where $\gamma \ge 0$. It follows that $(a_\lambda^2 + b_\lambda^2)=\lambda$, and 
either $b_\lambda=0$ or $(a_\lambda + c_\lambda)=0$. In the first case,
$\ket{\psi_\lambda}_m \ket 0_{m'_1}$ is an eigenvector of $H'$ with  eigenvalue
$\pm \sqrt \lambda$. In the second case, $\gamma=\lambda$ and
$H'$ has two eigenvectors with distinct eigenvalues $\pm \sqrt \lambda$.
Thus, in general, $\ket{\psi_\lambda}_m \ket 0_{m'_1} =\alpha_+\ket{\psi^+_\lambda}
+ \alpha_- \ket{\psi^-_\lambda}$, where $\ket{\psi^\pm_\lambda}$
are the eigenvectors of $H'$ of eigenvalues $\pm \sqrt \lambda$, respectively.
Let $W_t:=e^{-iH't}$ and 
\begin{align}
\tilde X_J : = \frac {\delta y}{\sqrt {2\pi}} \sum_{j=-J}^J e^{-y_j^2/2} W_{t_j} \; .
\end{align}
We obtain
\begin{align}
    \nonumber
    \tilde X_J & \ket{\psi_\lambda}_m  \ket 0_{m'_1}  = \alpha_+ \frac {\delta y}{\sqrt {2\pi}} \sum_{j=-J}^J e^{-y_j^2/2}
    e^{-it_j\sqrt{ \lambda}} \ket{\psi^+_\lambda} + \\
    &+ \alpha_- \frac {\delta y}{\sqrt {2 \pi}} \sum_{j=-J}^J e^{-y_j^2/2}
    e^{it_j\sqrt{ \lambda}}\ket{\psi^-_\lambda} \\
    & = \frac {\delta y}{\sqrt {2 \pi}} \sum_{j=-J}^J e^{-y_j^2/2}
    e^{-it_j\sqrt{ \lambda}} \ket{\psi_\lambda}_m \ket 0_{m'_1} \\
    \label{eq:XJaction}
    & = (X_J\ket{\psi_\lambda}_m) \ket 0_{m'_1}\;.
\end{align}
We used the property that the sums are invariant under the transformation
$y_j \rightarrow -y_j$, together with Eq.~\eqref{eq:X'Jaction}. 
Then,
\begin{align}
    \Tr[X_J] &= \sum_\lambda \bra {\psi_{\lambda}} X_J \ket {\psi_{\lambda}} \\
    & = \sum_\lambda \bra {\psi_{\lambda}} \bra 0_{m'_1} \tilde X_J \ket{\psi_{\lambda}} \ket 0_{m'_1}\\
    & = \Tr [\bra 0_{m'_1} \tilde X_J \ket 0_{m'_1}] \;. \label{eq:trX_J}
\end{align}

Let $W'_{t_j}$ be a unitary operator acting on $m+m'$ qubits that approximates $W_{t_j}$ as in Eq.~\eqref{eq:W'tcondition}.
Using $(2\delta y/\sqrt{2\pi})\sum_{j=1}^J e^{-y_j^2/2} \le 2$, we obtain
\begin{align}
\nonumber
 | &\Tr [\bra  0_{m'_1}  \tilde X_J \ket 0_{m'_1}]   -  \Tr [\bra 0_{m'}  X'_J \ket 0_{m'}] | \le ( {\delta y / \sqrt {2 \pi}}   )\\
& \times  \sum_{j=-J}^J  e^{-y_j^2/2}
| \Tr [\bra 0_{m'_1} W_{t_j} \ket 0_{m'_1}]   -  \Tr [\bra 0_{m'}  W'_{t_j} \ket 0_{m'}] | \\
& \le 2 ( {\delta y / \sqrt {2 \pi}} ) \sum_{j=1}^J  e^{-y_j^2/2} ( \eabs /8) \\
& \le  \eabs /4 \;.
\end{align}
Together with Eq.~\eqref{eq:trX_J}, this implies Eq.~\eqref{eq:tracerelation}.

\section{Estimation within relative error} 
\label{appx:est_rel_err}

	Our algorithm to estimate the partition function within given relative precision proceeds by making estimations within successively decreasing additive error. The intuition is that once the additive error becomes sufficiently small compared to the estimate obtained, the estimate is correct within a desired relative error. To get a correct relative estimate with high success probability,
	each additive estimation has to be done with decreasing probability of error as discussed below.
	We write $\cX$ for the quantity to estimate and assume $\cX_{\max} \ge \cX > 0$, for known $\cX_{\max}$. Let ${\textsc{Estimate}}(\eabs,\delta',\cX_{\max})$ be a procedure that outputs $\hat \cX$, satisfying $\cX_{\max} \ge \hat \cX \ge 0$
and
\begin{align}
	\Pr(|\hat \cX - \cX | \le \eabs ) \ge 1- \delta' \;,
\end{align}
where $\delta'>0$ is an upper bound on the probability of getting an estimate that is not within desired precision.

We claim that the following algorithm outputs an estimate
	$\hat \cX>0$ such that
	\begin{align}
	\label{eq:alg3output}
	\Pr(|\hat \cX - \cX | \le \erel \cX) \ge 1-\delta \;.
	\end{align}

	\begin{algorithm}[H]
		\renewcommand{\thealgorithm}{}
		\floatname{algorithm}{\hspace{3.3cm} Algorithm 3}
		\setstretch{1.0}
		\vspace{1mm}
		{\bf Input:} $\epsilon_{\rm rel}>0$, $\delta >0$, $\cX_{\max} > 0$. \\
		-- Set $r=0$, $\cX_0=\cX_{\max}$, and $\hat{\cX}_0=0$.\\
		-- While $\cX_r > \hat{\cX}_r$:\\
		\phantom{-- W} $r \leftarrow r+1$ .\\
		\phantom{-- W} Set $\cX_r = \cX_{\max}/2^r$, $\eabs(r) = \erel \cX_r/2$, and \\ 
		\phantom{-- W} $\delta'(r) = \tfrac{6}{\pi^2}\left(\delta/r^2\right)$.\\
		\phantom{-- W} $\hat{\cX}_r = {\textsc{Estimate}}(\eabs(r),\delta'(r),\cX_{\max})$ . \\
		{\bf Output:} $\hat \cX = \hat \cX_r$ .\\
		\vspace{-3mm}
		\caption{}
	\end{algorithm}

	We let $R \in \{1,2,\ldots \}$ be the number of times the procedure ${\textsc{Estimate}}$ is used until Algorithm 3 stops, i.e., the final value of $r$.
	This number is a random variable sampled with 
	probability $\Pr(R)$ and determines the complexity of the algorithm. Its expected
	value depends on $\cX$ and may be large if $\cX \ll \cX_{\max}$, as we discuss below. We also write $\Pr(\hat \cX)$ for the probability of obtaining output (i.e., relative estimate) $\hat \cX$.

    At any step $r$ of Algorithm 3, there is a non-zero probability
    of computing a wrong additive estimate, i.e., an estimate
    $\hat \cX_r$ such that $|\hat \cX_r - \cX| > \eabs(r)$; this can cause the algorithm to stop and report an undesired output. The overall probability of obtaining a correct output $\hat \cX$ then decreases with $r$ and, to bound it from below, we need to decrease $\delta'(r)$ with $r$. 
    
    For obtaining this bound, we consider an infinite sequence of independent estimates $\hat \cX_1,\hat \cX_2, \ldots$, where each $\hat \cX_r$ was obtained using ${\textsc{Estimate}}(\eabs(r),\delta'(r),\cX_{\max})$
    as in Algorithm 3. Let $R$ be the first position in this sequence at
    which $\hat \cX_r \ge \cX_r$ and $\hat \cX=\hat \cX_R$. Since the $\hat \cX_r$
    are uncorrelated, these $R$ and  $\hat \cX$ are the very same random variables as in Algorithm 3, with the same sampling probabilities $\Pr(R)$ and $\Pr(\hat \cX)$, respectively.
    Then, the probability of $\hat \cX$ being correct can be lower-bounded by the probability that every $\hat \cX_1,\hat \cX_2, \ldots$ in the infinite sequence is correct, i.e., the probability that $|\hat \cX_r - \cX| \le \eabs(r)$ for all $r$. This is simply $\prod_{r=1}^\infty \left(1-\delta'(r)\right)$, and we obtain
    \begin{align}
        \prod_{r=1}^\infty \left(1-\delta'(r)\right) &\ge 1 - \delta \left(\tfrac {6}{\pi^2}\right)
        \sum_{r=1}^\infty \frac 1 {r^2} \\
        \label{eq:Alg3confbound}
        & = 1-\delta\;. 
    \end{align}

    Then, Algorithm 3 succeeds and reports a correct outcome
    with probability at least $(1-\delta)$. In that case,
    the output 
    satisfies $|\hat \cX - \cX| \le \eabs(R)$, for the corresponding value of $R$. Since $\eabs(R)=\erel \cX_R/2$ and $\hat \cX \ge \cX_R$  in this case, then $ |\hat \cX - \cX| \le \erel \hat \cX/2$.
	Equivalently,
	\begin{align}
	(1-\erel/2)^{-1} \cX \ge \hat\cX \ge   (1+\erel/2)^{-1} \cX \;,
	\end{align}
	or
	\begin{align}
	| \hat \cX - \cX |& \le  \sum_{j \ge 1} (\erel/2)^j \cX \\ & 
	\label{eq:Alg4correct}
	\le \erel \cX\;.
	\end{align}
	As the condition in Eq.~\eqref{eq:Alg4correct} occurs with probability at least $(1-\delta)$,
	this proves Eq.~\eqref{eq:alg3output} and the correctness of Algorithm 3.

	The properties of $\Pr(R)$ are important
	for determining the complexity of Algorithm 3.
	As one might expect, we will show that this
	probability decays rapidly after a sufficiently
	high value of $R$.
	To this end, we let $q$ be an integer that
	satisfies $q+1 \ge \log_2(\cX_{\rm max}/\cX) + \log_2(3/2)> q$. It follows that $(3/2) \cX_{q+1} \le \cX < (3/2) \cX_q$ and if Algorithm 3 were to be ran
	with $\delta'(r)=0$, then it would stop at $R=q-1$, $R=q$,
	or $R=q+1$ with certainty because $\erel \le 1$.
	In reality, we need to choose $\delta'(r) > 0$ or otherwise the complexity
	of ${\textsc{Estimate}}$ may be large or diverge; however, we will show that $q$ determines the expected value of $R$
	for our choice of $\delta'(r)$, up to an additive constant.

	We now bound $\Pr(R\ge q')$, which is the cumulative probability that Algorithm 3 stops at a step $q'$ or later for $q'\ge q+2$. 
	This is the same as the probability of Algorithm 3 not stopping before $q'$. Since the estimates are uncorrelated, this is given by the product of the probabilities of not stopping at $r=1,2,\ldots,q'-1$. 
	In particular, we can upper bound the probabilities of not stopping at $r\le q$  by 1.
	Once $r \ge q+1$, we have $\cX \ge (3/2) \cX_{q+1} \ge \cX_{r}+ \eabs(r)$ and the estimate $\hat \cX_{r}$ is greater or equal than $\cX_{r}$ with probability at least $(1-\delta'(r))$. Then, the probability of not stopping at this $r$ is at most $\delta'(r)$ and we obtain
	\begin{align}
	    \Pr(R\ge q') &\le \prod_{r=q+1}^{q'-1} \delta'(r)\\
	    &= \left( \tfrac {6 } {\pi^2} \right)^{q'-q-1} \prod_{r=q+1}^{q'-1} \frac{1}{r^2} \\
    	& = \left( \tfrac 6 {\pi^2} \right)^{q'-q-1} \left( \frac{q!}{(q'-1)!}\right)^2 \\
    	& \le \left( \tfrac 6 {\pi^2} \right)^{q'-q-1} \frac 1 {(q'-q-1)!} \;. \label{eq:probsuperexpdecay}
	\end{align}
	Therefore, the probability that Algorithm 3 stops at 
	$R \ge q+2$ decays super-exponentially.
	
    We can use this property of $\Pr(R)$ to bound the expected value of $R$ as follows:
    \begin{align} 
        E_R &=\sum_{R=1}^{\infty} \Pr(R)\, R \\
        & = \sum_{q'=1}^\infty \Pr(R\ge q') \\
        &= \sum_{q'=1}^{q+1} \Pr(R\ge q')+\sum_{q'=q+2}^\infty \Pr(R\ge q') \\
        & \le (q+1) + \sum_{q'=q+2}^\infty \Pr(R\ge q')\;.\label{eq:exptosumprob}
    \end{align}
   Equation~\eqref{eq:probsuperexpdecay} then implies
    \begin{align}
    	\sum_{q'=q+2}^{\infty}&\Pr(R \ge q') \nonumber \\
    	&\le \sum_{q'=q+2}^{\infty}
    	\left( \tfrac 6 {\pi^2} \right)^{q'-q-1} \frac 1
    	{(q'-q-1)!} \\
     & = e^{6/\pi^2}-1 \\
     & < 1 \;,
    \end{align}
    	and hence
    \begin{align}
    	E_R &< q+2 \\
    	& < \log_2(\cX_{\rm max}/\cX) + \log_2(3/2) + 2 \\
    	& < \log_2(\cX_{\rm max}/\cX) +3\;.
    \end{align}
    The expected running time of Algorithm 3,  which is given by $E_R$, is then $O(\log_2(\cX_{\max}/\cX))$. As defined,
    this algorithm can run indefinitely but the 
    properties of $\Pr(R)$ make it unlikely that we will obtain running times $R$ that are much larger
    than $E_R$.

    \section{Complexity of relative-error estimation}
    \label{app:rel_err_complexity}
    
    We now bound the expected complexity of Algorithm 2 for the PFP. The complexity of $\textsc{EstimatePF}$ is different depending on whether we use Algorithm 1.A or 1.B. 
   Disregarding polylogarithmic factors in $\beta$ and $m$, the complexities of Algorithms 1.A and 1.B are
  \begin{align}
    \label{eq:1Acompapp}
        T_{1A} = \tilde O \Bigg(&\frac{e^{2\beta} M^2}{\eabs^2}\log(1/\delta)L^3C_H\nonumber \\ &\times (\beta^2+m^2)\left(\log\left(1/\eabs\right)\right)^2\Bigg) \;
    \end{align}
    and
    \begin{align}
    \label{eq:1Bcompapp}
        T_{1B}= \tilde O \Bigg(& \frac{ M^2}{\eabs^2} \log(1/\delta){L^3 C_H }\nonumber \\
        &\times m(\beta+m)\left(\log\left( 1/\eabs\right)\right)^2 \Bigg) \;.
    \end{align}
     We may write the cost of running $ \textsc{EstimatePF} $ once at a step $r$ in Algorithm 2 in the general form
     if we replace $\eabs \rightarrow \erel \cZ_{\max}/2^{r+1}$ and $\delta \rightarrow (6/\pi^2) \delta/r^2$. Then,
    \begin{align}
        T_1(r) = \tilde O \left( \frac {\gamma \log\left(1/\delta \right)}{\erel^2 \cZ_{\max}^2} 4^{r}r^2 \log(r) \right) \;,
    \end{align}
    where the $\tilde O$ notation hides polylogarithmic factors in $\erel$. The factor $\gamma$ depends on $\beta$, $m$, $L$ and $C_H$, and is different for Algorithms 1.A and 1.B. 
    
    Then, the overall cost of running $\textsc{EstimatePF}$ $R$ times in Algorithm 2 is
    \begin{align}
       T_2(R) & = \sum_{r=1}^R T_1(r) \\
       & = \tilde O \left(\frac {\gamma\log\left(1/\delta\right)}{\erel^2 \cZ_{\max}^2} \sum_{r=1}^R 4^r r^2 \log(r) \right) \\
       & = \tilde O \left( \frac {\gamma\log\left(1/\delta\right)}{\erel^2 \cZ_{\max}^2}  4^R R^3 \right) \; ,
    \end{align}
     where the $\tilde O$ notation hides a polylogarithmic factor in $\erel$.
    We can compute the expected value of the cost of Algorithm 2 as $T_2 =\sum_{R=1}^\infty \Pr(R) T_2(R)$. We note that
    \begin{align}
    \sum_{R=1}^\infty & \Pr(R) 4^R R^3 \nonumber \\
    &=\sum_{R=1}^{q+1} \Pr(R) 4^R R^3  + \sum_{R=q+2}^{\infty}  \Pr(R) 4^R R^3 \\
    & \le 4^{q+1}(q+1)^3 + \sum_{q'=q+2}^\infty \Pr(R \ge q') 4^{q'} q'^3 \;.
    \end{align}	
    We use Equation~\eqref{eq:probsuperexpdecay} to bound the second term in the above and obtain
    \begin{align}
        \sum_{q'=q+2}^\infty & \Pr(R \ge q')  4^{q'} q'^3  \nonumber \\ 
        &\leq \sum_{q'=q+2}^\infty \left( \tfrac 6 {\pi^2} \right)^{q'-q-1} \frac {4^{q'} q'^3} {(q'-q-1)!} \\
        & = 4^{q+1}\sum_{q''=0}^\infty \left(\tfrac{24}{\pi^2 }\right)^{q''} \frac {(q''+q+1)^3} {q''!}\\
        & = O(4^q q^3)\;.
    \end{align}
    Therefore, the expected complexity of Algorithm 2 can be written as
    \begin{align}
        T_2 &= \tilde O \left( \frac {\gamma\log\left(1/\delta\right)}{\erel^2 \cZ_{\max}^2} 4^q q^3 \right) \nonumber \\
        &= \tilde O \left(\frac {\gamma\log\left(1/\delta\right)}{\erel^2 \cZ^2} \left(\log\left(\cZ_{\max}/\cZ\right)\right)^3 \right)\; .
    \end{align}
    For Algorithm 1.A, we have $\gamma = e^{2\beta} M^2 {L^3 C_H }(\beta^2+m^2) $ and $\cZ_{\max}=Me^\beta$, giving
    \begin{align}\label{eq:complexity2A_det}
        T_{2A} =\tilde O \Bigg(&\left(\frac{Me^{\beta}}{\erel \cZ}\right)^2 {L^3 C_H }\log\left(1/\delta\right)\nonumber \\ &\times(\beta^2+m^2)\left(\log\left(Me^\beta/\cZ\right)\right)^3 \Bigg)\;.
    \end{align}
    For Algorithm 1.B, we have $\gamma = M^2 {L^3 C_H }m(\beta+m) $ and $\cZ_{\max}=M$, giving
    \begin{align}\label{eq:complexity2B_det}
        T_{2B} = \tilde O \Bigg(&\left(\frac{M}{\erel \cZ}\right)^2 {L^3 C_H }\log\left(1/\delta\right) \nonumber\\ &\times m(\beta+m)\left(\log\left(M/\cZ\right)\right)^3 \Bigg)\;.
    \end{align}
  The $\tilde O$ notation in $T_{2A}$ and $T_{2B}$ hides polylogarithmic factors in $\beta$, $m$, and $\erel$.

\bibliographystyle{ieeetr}
\bibliography{PartitionFunctionDQC1}

\end{document}